\documentclass[letterpaper]{article} %
\usepackage[]{aaai2026}  %
\usepackage{times}  %
\usepackage{helvet}  %
\usepackage{courier}  %
\usepackage[hyphens]{url}  %
\usepackage{graphicx} %
\usepackage{natbib}  %
\usepackage{caption} %
\usepackage{algorithm}
\usepackage{algorithmicx}

\usepackage{newfloat}
\usepackage{listings}
\DeclareCaptionStyle{ruled}{labelfont=normalfont,labelsep=colon,strut=off} %
\floatstyle{ruled}
\newfloat{listing}{tb}{lst}{}
\floatname{listing}{Listing}
\usepackage{graphicx} %
\usepackage{amsmath}
\usepackage{amsthm}
\usepackage{natbib}
\usepackage{booktabs}
\usepackage{algorithm}
\usepackage{algpseudocode}
\usepackage{amsthm}
\usepackage{colonequals}
\usepackage{dsfont}
\usepackage{pgfplots}
\usepackage{filecontents}
\usetikzlibrary{arrows}
\usetikzlibrary{shapes.geometric}
\usepgfplotslibrary{fillbetween}
\tikzstyle{arrow} = [thick,->,>=stealth]
\usepackage{adjustbox}

\usepackage{mathtools}

\usepackage{thmtools,thm-restate}
\newtheorem{obs}{Observation}
\newtheorem{lem}{Lemma}
\newtheorem{thm}{Theorem}

\title{Intermediate N-Gramming: Deterministic and Fast N-Grams For Large N and Large Datasets}
\author{
    Ryan R. Curtin\textsuperscript{\rm 1},
    Fred Lu\textsuperscript{\rm 1,2},
    Edward Raff\textsuperscript{\rm 2,3},
    Priyanka Ranade\textsuperscript{\rm 4}
}
\affiliations{
    \textsuperscript{\rm 1}Booz Allen Hamilton\\
    \textsuperscript{\rm 2}University of Maryland, Baltimore County\\
    \textsuperscript{\rm 3}CrowdStrike\\
    \textsuperscript{\rm 4}Laboratory for Physical Sciences\\
}

\begin{document}

\maketitle

\begin{abstract}
The number of $n$-gram features grows exponentially in $n$, making it computationally demanding to compute the most frequent $n$-grams even for $n$ as small as $3$.
Motivated by our production machine learning system built on $n$-gram features, we ask:
is it possible to accurately, deterministically, and quickly recover the top-$k$ most frequent $n$-grams?
We devise a multi-pass algorithm called {\it Intergrams} that constructs candidate $n$-grams from the preceding $(n-1)$-grams.
By designing this algorithm with hardware in mind,
our approach yields more than an order of magnitude speedup (up to 33$\times$!) over the next known fastest algorithm,
even when similar optimization are applied to the other algorithm.
Using the empirical power-law distribution over n-grams, we also provide theory to inform the efficacy of our multi-pass approach.
Our code is available at \small{\tt https://github.com/rcurtin/Intergrams}.
\end{abstract}

\section{Introduction}
\label{sec:intro}

The goal of this work is the problem statement: ``\textit{I should be able to  find the top-$k$ $n$-grams of a dataset as quickly as the disk can give me bytes}''.
If one has a large set of files
(TBs or larger),
and wants to find the $k$ most frequent
subsequences of length $n$ in that set of files (called {\em $n$-grams}),
the claim is that the bottleneck in computation should be
getting the bytes from the files---{\em not} processing them.
This problem is {\em much} harder than it looks from the description,
but it {\em is} possible to satisfy the claim.

The problem is of interest
not because it was lunchtime banter (although it was),
but instead 
because it is an important preprocessing step for
numerous real-world machine learning tasks:
$n$-grams are effective features as inputs to machine learning models.
Once the top-$k$ $n$-grams are known,
it is extremely fast to convert a large set of files into
sparse data that can then be directly used for training---see Fig.~\ref{fig:pipeline}.
But,
for large datasets, $k$, or $n$,
the time to compute the top-$k$ $n$-grams can
significantly exceed all other steps.

Probably the most well-known application of this strategy,
is malware classification:
byte-sequence features are commonly used
to produce simple, interpretable models
that give state-of-the-art performance
and extremely fast inference times~\citep{masud2008scalable, fuyong2017malware,raff2018hash}.
For very large datasets,
they may even be used for distributed training of
models~\citep{lu2025optimizing,lu2024high}.
For text processing,
length-$n$ character sequences
have been shown to outperform text embedding models for (certain) textual event classification~\citep{piskorski2020tf}
and are useful in large-scale language modeling~\citep{liu2024infini}.
Finally, $n$-grams of genomic sequences
(known as $k$-mers)
have proven to be useful for various tasks,
from phylogenetic analysis to species classification~\citep{chor2009genomic,yang2020intrinsic}.

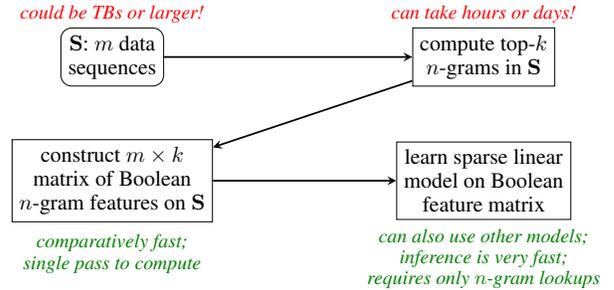
\begin{figure}[t]
\begin{center}
\begin{adjustbox}{width=0.95\columnwidth}
    \begin{tikzpicture}[auto, node distance=4cm,>=latex]

    \tikzstyle{block} = [draw, rectangle, align=center];
    \tikzstyle{rblock}=[draw, shape=rectangle,rounded corners=0.5em, align=center];

    \node[rblock] (data) { $\mathbf{S}$: $m$ data\\ sequences };
    
    \node[block,right of=data, xshift=2cm] (ngram) { compute top-$k$\\ $n$-grams in $\mathbf{S}$ };

    \node[block,below of=data, yshift=2cm] (feature) { construct $m \times k$\\ matrix of Boolean\\ $n$-gram features on $\mathbf{S}$ };

    \node[block,below of=ngram, yshift=2cm] (ml) { learn sparse linear\\ model on Boolean\\ feature matrix };

    \draw [arrow] (data) -- (ngram);
    \draw [arrow] (ngram) -- (feature);
    \draw [arrow] (feature) -- (ml);

    \node [above of=data,yshift=-3.3cm] () { \color{red}{\it \footnotesize could be TBs or larger!} };
    \node [above of=ngram,yshift=-3.3cm] () { \color{red}{\it \footnotesize can take hours or days!} };
    \node [below of=feature,yshift=3.0cm] () { \color{green!50!black}{\it \footnotesize comparatively fast;} };
    \node [below of=feature,yshift=2.65cm] () { \color{green!50!black}{\it \footnotesize single pass to compute} };
    \node [below of=ml,yshift=3.1cm] () { \color{green!50!black}{\it \footnotesize can also use other models;} };
    \node [below of=ml,yshift=2.75cm] () { \color{green!50!black}{\it \footnotesize inference is very fast;} };
    \node [below of=ml,yshift=2.4cm] () { \color{green!50!black}{\it \footnotesize requires only $n$-gram lookups} };
\end{tikzpicture}
\end{adjustbox}
\end{center}
\vspace*{-1.5em}
\caption{A typical machine learning pipeline using $n$-grams as features.
To run this pipeline, the $k$ most common $n$-grams from the input data must be known.
Often, on large datasets, the process of computing the top-$k$ $n$-grams is
far more expensive than the modeling step!}
\label{fig:pipeline}
\end{figure}

\section{Problem Formalization and Notation}
\label{sec:problem}

The problem of computing the top-$k$ $n$-grams is straightforward to formalize:
suppose that there is a dictionary $\mathbf{D}$ of possible $n$-grams;
this dictionary can be extremely large.
If $n$-grams are byte sequences, then $|\mathbf{D}| = 256^n$
(e.g. there are 18 quintillion possible $8$-grams!).
Suppose also we are given a dataset that consists of $m$
sequences
$\mathbf{S} = \{ s_1, s_2, \ldots, s_m \}$.
In a malware detection application, for instance,
$\mathbf{S}$ may be a large collection of files
where each $s_i$ is one file.
Our task, then, is to find one of the following two quantities:

\begin{eqnarray}
k\operatorname{-argmax}_{x \in \mathbf{D}} \sum_{s_i \in \mathbf{S}} \sum_{x_j \in s_i} \mathds{1}(x_j = x), \label{eqn:count_all} \\
k\operatorname{-argmax}_{x \in \mathbf{D}} \sum_{s_i \in \mathbf{S}} \mathds{1}(x \in s_i). \label{eqn:count_one}
\end{eqnarray}

In the first variant (Eq.~\ref{eqn:count_all}),
each $n$-gram is counted once {\em every} time that it appears.
In the second variant (Eq.~\ref{eqn:count_one}),
an $n$-gram is only counted once for each sequence $s_i$ that it appears in.
When using the top-$k$ $n$-grams as features for a machine learning model,
especially on real-world data,
it is generally more effective to count an $n$-gram once per sequence.
This is because it is often the case that a single $n$-gram
may be present in only a very few sequences of $\mathbf{S}$,
but may occur very many times in those few sequences.
If our goal is to predict the label of a sequence,
then these $n$-grams carry little predictive value,
as their feature values will be $0$ for the vast majority of sequences in $\mathbf{S}$.
For more details on this phenomenon,
see the discussion in~\citet{raff2018investigation}.

As a result,
our focus in this paper will be on solving Eq.~\ref{eqn:count_one};
however, our algorithms can be readily adapted to solve Eq.~\ref{eqn:count_all}
(and they will solve that problem more quickly as it is an easier problem).

\section{Related Work}
\label{sec:related}

The naive approach to solving this problem
consists of storing a large dictionary for each element in $\mathbf{D}$
and iterating over each sequence $s_i$,
incrementing the appropriate counts in the dictionary.
Then, when all sequences have been iterated over,
sort the dictionary by count and keep the top $k$.

Such a strategy can be implemented in only a few lines of code;
in fact, this is what is done by the commonly-used
{\tt CountVectorizer} class from the {\tt scikit-learn} package~\citep{pedregosa2011scikit}.
But this obviously scales horrendously for non-trivial sizes of $\mathbf{D}$ or $k$.

A number of stream-based approaches have been developed, e.g.~\citet{charikar2002finding}, \citet{cormode2005summarizing}, and \citet{jin2003dynamically}.
These approaches build a `sketch' of the data as the stream is processed,
discarding $n$-grams that are found to be infrequent.
Although these approaches perform better than the naive strategy described above,
they tend to focus on much smaller $k$ than would be interesting for a machine learning use case,
and require storing significantly more than $k$ partial counts.
For instance, the Space-Saving algorithm of~\citet{metwally2005efficient}
requires a number of counters $M >> k$ for error bounds to be acceptably small.
Further, adapting these algorithms to count an $n$-gram only once per sequence (Eq.~\ref{eqn:count_one})
adds additional overhead.

A more principled approach can be developed by observing that
the $n$-gram distribution of real-world data
is {\it not} uniform
but instead follows a Zipfian (power-law) distribution~\citep{raff2018investigation}.
That is, few $n$-grams appear regularly;
the vast majority of $n$-grams either appear rarely
or not at all
(and as such have no chance of being in the top $k$).

Using this reality, \citet{raff2018hash} proposed the `hash-gramming' approach,
which is the current best-known algorithm for
finding top-$k$ $n$-grams for non-trivial $k$.
This approach computes the hash of each $n$-gram,
mapping it from a space of size $|\mathbf{D}|$
to a more manageable number of elements;
in their implementation, they hash to roughly $2^{31}$ elements.
Processing is done in two passes:
in the first pass,
the hashes of each $n$-gram are counted.
Then, the top-$k$ hashes are computed.
In the second pass, exact $n$-gram values using an approach like the naive dictionary approach,
but {\em only} for $n$-grams whose hashes are in the set of top-$k$ hashes.
This reduces the size of memory needed for the second pass to roughly $O(k)$ memory,
which is far faster and more manageable than $O(|\mathbf{D}|)$ memory!
With a good implementation,
the hash-gramming approach can yield throughput of $10-20$ MB/s;
although this is orders of magnitude faster than both the naive dictionary algorithm
and the Space-Saving algorithm,
it is still far away from the `disk speed' postulated in the introduction
(which would be more like $600$ MB/s for a typical desktop computer with a spinning-disk drive).

\section{Hardware Limitations}
\label{sec:hardware}

Although the algorithms described in the previous section focus on
minimizing overall memory usage
they do not focus on actual performance limitations of physical hardware.
The biggest problem that all of the algorithms described in the previous section face is
{\bf non-contiguous memory access patterns} into {\bf large blocks of memory}.

On a modern computer,
data structures that cannot fit into the processor caches (L1/L2/L3)
are instead stored in RAM.
Each level of this memory hierarchy has orders-of-magnitude different
capacity,
bandwidth,
and latency.
Table~\ref{tab:memory} shows typical figures on modern hardware
as provided by Intel~\cite{intel_memory},
including for disk speed.

\begin{table}[t!]
\begin{center}
\begin{adjustbox}{width=0.99\columnwidth}
\begin{tabular}{cccc}
\toprule
{\bf Hierarchy level} & {\bf Capacity} & {\bf Peak Bandwidth} & {\bf Latency} \\
\midrule
L1 cache        & 16 KB--64 KB  & 1 TB+/s          & $\sim1$ ns \\
L2 cache        & 64 KB--512 KB & 1 TB+/s          & $\sim4$ ns \\
L3 cache        & 8 MB--256 MB  & 100-500 GB/s     & $\sim40$ ns \\
RAM             & 1 GB--4 TB    & 10-100 GB/s      & $\sim80$ ns \\
Disk            & 1 TB+         & 500 MB/s--5 GB/s & $\sim8-80$ $\mu$s \\
\bottomrule
\end{tabular}
\end{adjustbox}
\end{center}
\vspace*{-0.8em}
\caption{Typical performance characteristics of each level of the memory hierarchy on modern computers~\cite{intel_memory}.
Whenever data is too large to fit in a cache level,
access times become significantly longer,
and bandwidth decreases,
both by orders of magnitude.}
\label{tab:memory}
\end{table}

An important reality somewhat masked by Table~\ref{tab:memory} is that
peak bandwidth numbers are for sequential accesses;
random non-sequential access into a block of memory can be orders of magnitude slower.
Thus, for some array {\small \tt a},
a loop that modifies {\small \tt a[0]}, then {\small \tt a[1]}, then {\small \tt a[2]},
and so forth
is much faster than a loop that
modifies random elements of {\small \tt a}.
See Figure~\ref{fig:mem_bw} for a simulation that shows this effect
(imitating the excellent work of~\citet{drepper2007every}).

\begin{figure}[t!]
\begin{adjustbox}{width=0.99\columnwidth}
    \begin{tikzpicture}
\begin{axis}[
      axis lines=left,
      width=0.55\textwidth,
      height=0.26\textwidth,
      ymin=0.1,ymax=100,
      xmin=4096, xmax=68719476736,
      xlabel={Array size},
      xmode=log,
      log basis x={2},
      xtick={4096, 65536, 1048576, 16777216, 268435456, 4294967296, 68719476736},
      xticklabels={4 KB, 64 KB, 1 MB, 16 MB, 256 MB, 4 GB, 64 GB},
      ymode=log,
      log basis y={10},
      yticklabels={0, 100 MB/s, 1 GB/s, 10 GB/s, 100 GB/s},
      legend pos=south west,
      grid=both,grid style={line width=0.6pt, draw=gray!10}
]
    \addplot table [x=bytes, y=seqbw, col sep=comma, color=blue] {data/mem_bw.csv};
    \addlegendentry { Sequential };
    \addplot table [x=bytes, y=rndbw, col sep=comma, color=red] {data/mem_bw.csv};
    \addlegendentry { Random };

    \addplot [name path=sequpper,draw=none] table[x=bytes, y=seqlobw, col sep=comma] {data/mem_bw.csv};
    \addplot [name path=seqlower,draw=none] table[x=bytes, y=seqhibw, col sep=comma] {data/mem_bw.csv};
    \addplot [fill=blue!30] fill between[of=sequpper and seqlower];
    
    \addplot [name path=rndupper,draw=none] table[x=bytes, y=rndlobw, col sep=comma] {data/mem_bw.csv};
    \addplot [name path=rndlower,draw=none] table[x=bytes, y=rndhibw, col sep=comma] {data/mem_bw.csv};
    \addplot [fill=red!30] fill between[of=rndupper and rndlower];
\end{axis}
\end{tikzpicture}
\end{adjustbox}
\vspace*{-1em}
\caption{Memory bandwidth benchmarks.
We vary the size of the array
(x axis),
and in each trial we increment all elements of the array in sequential (blue) or random (red) order.
For sequential access patterns, the hardware prefetcher can keep bandwidth high,
but the 2048-element TLB necessarily starts to miss at 8 MB (we use 4 KB page sizes).
For random access patterns, the hardware prefetcher
cannot predict the memory (or page) that is needed next,
and throughput craters as the array size exceeds the size of each level in the cache hierarchy.}
\label{fig:mem_bw}
\end{figure}
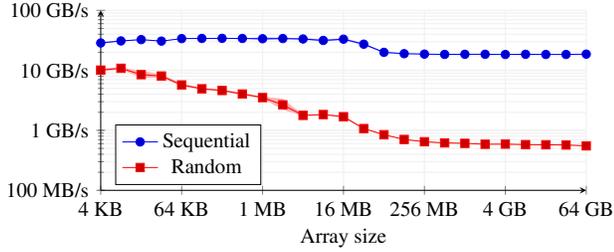

Another important reality of hardware is paging and virtual memory:
when we access a random location in memory,
the operating system must first map the given memory pointer
to the true underlying hardware location
using the virtual memory page tables.
Virtual memory is organized into pages (typically of size 4 KB but this can vary)
and therefore the first step in any virtual memory access
is to get the underlying hardware location of the page.
The processor's translation lookaside buffer (TLB)
assists with this process
by caching entries of recently-used pages.
But the TLB's capacity is limited;
typical modern TLBs may have only hundreds of entries.
Whenever a page is not found in the TLB,
the kernel must perform a time-consuming page walk;
this is seen in Fig.~\ref{fig:mem_bw}.

In order to design a fast algorithm for $n$-gramming,
we must be aware of these hardware concerns.
Similar to the `cache-oblivious' algorithms pioneered by~\citet{frigo1999cache} and~\citet{demaine2002cache},
we will use a simple model of a cache
and then ensure that our algorithm operates in a way that minimizes accesses outside the cache.

\section{Hash-gramming Memory Problems}
\label{sec:hashgramming}

Existing algorithms for fast top-$k$ $n$-gram computation
severely violate the understanding of the memory hierarchy we established above.
The working sets of these algorithms often significantly exceed cache size,
and the memory access patterns are effectively random.
Although the analysis is similar for the Space-Saving algorithm and others,
we focus here on Hash-Gramming,
given in Algorithm~\ref{alg:hash_grams},
as it is the fastest competitor to our algorithm by orders of magnitude
and its shortcomings served heavily as inspiration.

Hash-Gramming operates by taking two passes on the data:
in the first pass,
a hash function is used to map an $n$-gram into one of $B$ buckets.
Then, the top-$k$ buckets are retained,
and in the second pass,
any $n$-grams whose hashes are in the top-$k$ buckets
are counted;
finally, the top-$k$ $n$-grams of the second pass are returned.
Because of the condition that an $n$-gram hashes to the top-$k$ buckets,
and $k \ll B$,
the second pass is significantly faster.

\begin{algorithm}[t]
\small
\begin{center}
\begin{algorithmic}[1]
\Require Bucket size $B$, hash function $h(\cdot)$, sequence set $\mathbf{S}$, $n$, $k$

\medskip

\State $T \gets \mathbf{0}^{B}$
\Comment{First pass: compute $n$-gram hash counts.}
\ForAll{sequence $s_i$ in $\mathbf{S}$ {\bf in parallel}}
    \State $T_i \gets \mathbf{0}^{B}$
    \ForAll{$n$-gram $x_i \in S$}
        \State $T_i[h(x_i)] \gets 1$
    \EndFor
    \State $T \gets T + T_i$
\EndFor

\medskip
\State $H \gets k\operatorname{-argmax}_{h_i \in T} T[h_i]$
\Comment{Top-$k$ hash values.}

\medskip
\State $D \gets \mathrm{empty\ dictionary}$
\Comment{Second pass: compute exact counts of $n$-grams.}
\ForAll{Sequence $s_i$ in $\mathbf{S}$ {\bf in parallel}}
    \State $D_i \gets \mathrm{empty\ dictionary}$
    \ForAll{$n$-gram $x_i \in s_i$}
        \If{$h(x_i) \in H$}
            \State $D_i[x_i] \gets 1$
        \EndIf
    \EndFor
    \State $D \gets D + D_i$
\EndFor

\medskip
\noindent \Return $k\operatorname{-argmax}_{x_i \in D} D[x_i]$ \Comment{Return final top-$k$ $n$-grams.}
\end{algorithmic}
\end{center}
\caption{Hash-Gramming~\citet{Kilograms_2019},
specialized to count an $n$-gram once per sequence.}
\label{alg:hash_grams}
\end{algorithm}

But let us consider the memory access behavior of this algorithm:
at each step in the first pass,
we compute the hash of an $n$-gram, $h(x_i)$.
$h(\cdot)$ must hash $n$-grams randomly to one of $B$ buckets,
but there is no correlation between the buckets that
two consecutive $n$-grams will hash to.

Thus, each update to $T_i$ is effectively random.
Problematically, $B$ is large: typically around $2^{31}$.
This means that if $T_i$ is represented as a bit vector,
its size is $\sim\!\!8$ MB.
Although this fits into most L3 caches in the single-threaded case,
with many threads operating on different sequences,
the sum total of all bit vectors $T_i$
can quickly exceed L3 cache.

Worse, the update step for $T$ also exhibits random access behavior:
because of the Zipfian distribution generally seen in real-world data,
each $T_i$ is highly sparse
(as most $n$-grams never appear).
The throughput for this update step is very low,
because the size of $T$ is typically $8$ or $16$ GB,
depending on the precision used.
Even when $T_i$ is scanned sequentially,
the sparse updates to $T$
cannot benefit from hardware prefetching
and often suffer from page faults.

The second pass also exhibits the same random memory access patterns,
but because the size of $D$ is so much smaller,
and because many $n$-grams are skipped since they are not in $H$,
the observed runtime effect is comparatively less.

Although there are a number of small tricks that can be used to
hide the latency and low bandwidth associated with the hash-gramming approach,
these do not change the fundamental underlying access patterns.
The use of a small LRU cache or Cuckoo hash table~\citep{pagh2004cuckoo}
to cache common recently-seen hash values $h(x_i)$
and avoid lookups in or updates to $T_i$ or flushes to $T$
can be helpful,
but again the Zipfian distribution of data hurts:
due to the long tail of the distribution,
the vast majority of $n$-grams are only seen once or twice
and will not remain in the small cache,
requiring accesses to $T_i$ and $T$.
Software prefetching with a circular buffer presents another alluring acceleration,
but on many processors,
a TLB miss on a prefetch will cause an immediate stall and page walk~\citep{intel_reference_optimization}.
Given the large size of $T$, these TLB misses are almost guaranteed with a page size of $4$ KB (the default).
It is possible to use large pages of $2$ MB on typical Linux systems\footnote{Linux kernel support for $1$ GB pages is available, but not everywhere, and is tedious to reliably build into a distributable software product for a variety of non-user-friendly reasons, so we do not consider it in our work.  As with $2$ MB pages, it would provide a small speedup but not address the underlying issue.},
but while this provides a speedup by reducing the number of TLB misses,
the size of the TLB for $2$ MB pages is generally smaller,
and it does not address the underlying issue.

\section{Warmup: Fast $3$-grams}
\label{sec:threegrams}

To achieve fast speeds for $n$-gramming,
we {\em must} address the memory access patterns
and avoid random accesses into large data structures that
cannot fit in processor caches
and spill across many memory pages.
We have two desiderata:

\begin{enumerate}
    \item If we must do random access into a block of memory, that block of memory should fit in the processor cache.

    \item If a block of memory must be larger than the processor cache, then we must access it sequentially.
\end{enumerate}

When we are passing over all the sequences in $\mathbf{S}$,
we will necessarily be visiting $n$-grams in effectively random order.
Thus, when we take a pass over the data,
we cannot expect to increment counts or any intermediate data structure
in a sequential manner.
To satisfy our first requirement, then,
\textit{our intermediate data structure must fit in the processor cache.}

Consider what would happen if we were only interested in computing the top-$k$ $3$-grams (e.g. $n = 3$).
If we are using byte-sequence data, the number of possible $3$-grams (that is, $|\mathbf{D}|$)
is $256^3 = 16777216$.
Then, for each sequence $s_i$,
we can hold a bit vector with one bit for each of the $16$M possible $3$-grams,
which has overall size $2$ MB: this trivially fits into the processor cache.
Since we cannot assume any ordering of the $n$-grams in the sequences $s_i$,
our memory access pattern into this bit vector is effectively random.

After processing the sequence $s_i$,
we can flush this bit vector to a global array of $16$M elements
using a sequential access pattern.
When using 32-bit integers, this has size $64$ MB, which still fits comfortably in cache (L3).
See Algorithm~\ref{alg:gram3_fast} for a complete description of this algorithm.

\begin{algorithm}[t]
\small
\begin{center}
\begin{algorithmic}[1]
\Require sequence set $\mathbf{S}$, $k$

\medskip

\State $C \gets \mathbf{0}^{16777216}$
\Comment{Size: 64 MB.}
\ForAll{sequence $s_i$ in $\mathbf{S}$ {\bf in parallel}}
    \State $C_i \gets \mathbf{0}^{16777216}$
    \Comment{Bit vector of size 2 MB.}
    \ForAll{$3$-gram $x \in S_i$}
        \State $C_i[x] \gets 1$
        \Comment{Random access into bit vector.}
    \EndFor
    \State $C \gets C + C_i$ \Comment{Sequential flush.}
\EndFor

\medskip
\noindent \State {\bf return} $k\operatorname{-argmax}_{x_i \in C} C[x_i]$
\Comment{Compute top-$k$ $3$-grams.}
\end{algorithmic}
\end{center}\caption{A fast algorithm to compute the top-$k$ $3$-grams on byte sequences, counting an $n$-gram once per sequence.}
\label{alg:gram3_fast}
\end{algorithm}

This algorithm is embarrassingly parallel;
each sequence can be handled by a different processor,
with the only additional memory requirement that each processor needs its own $2$ MB bit vector.
For most processors, the sum total of memory required still fits easily in cache.
On modern commodity hardware,
this algorithm can
compute $n$-grams as quickly as input sequences can be provided from disk.

There are additional optimizations we can make.
First, we can collect bit vectors $C_i$ to flush simultaneously:
instead of flushing a $C_i$ as soon as a sequence is complete,
we wait until a specific number of bit vectors are ready, and then flush them all simultaneously:
$C \gets C + C_{i} + C_{i + 1} + \ldots + C_{i + 8}$.
Second, we can accelerate each individual flush via the use of SIMD instructions.
Processors that support AVX2 are able to simultaneously add 8 integers in a single instruction.
In our tuned implementation, we are able to increment $8$ elements of $C$ for a single bit vector $C_i$ with
only 5 SIMD instructions.
On the powerful system we use for our experiments, this achieves peak disk bandwidth ($\sim\!5$ GB/s!).

\begin{figure}[b!]
\begin{center}
\begin{tikzpicture}
\begin{axis}[
      axis lines=left,
      width=0.45\textwidth,
      height=0.2\textwidth,
      ymin=95,ymax=100,
      xmin=1.0, xmax=3.0,
      xlabel={Oversample factor $z$},
      yticklabel={$\pgfmathprintnumber{\tick}$\%},
      grid=both,grid style={line width=0.6pt, draw=gray!10}
]

    \addplot table [x=z, y=pct, col sep=comma, color=red, mark=none] {data/ember2018_top_k_pct_10k.csv};
    \addlegendentry { k=$10$k };
    
    \addplot table [x=z, y=pct, col sep=comma, color=blue, mark=none] {data/ember2018_top_k_pct_100k.csv};
    \addlegendentry { k=$100$k };
    
    \addplot table [x=z, y=pct, col sep=comma, color=green, mark=none] {data/ember2018_top_k_pct_1M.csv};
    \addlegendentry { k=$1$M };
\end{axis}
\end{tikzpicture}
\end{center}
\vspace*{-0.8em}
\caption{Percentage of top-$k$ 4-grams that have 3-byte prefixes in the top-$zk$ 3-grams on the EMBER dataset,
as a function of $z$.
100\% means that {\em all} 4-grams have 3-byte prefixes in the top-$zk$ 3-grams;
in this case, we can filter the set of possible 4-grams significantly,
with no loss of accuracy!}
\label{fig:topk-pct}
\end{figure}
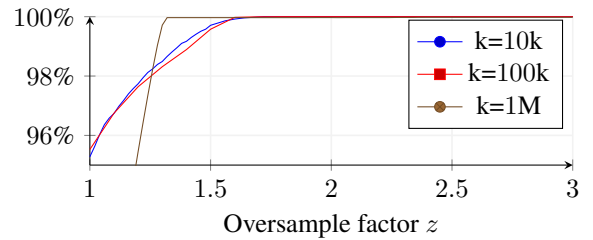

Now, what utility do $3$-grams have when we are interested in general $n$-grams for $n > 3$?
Intuitively, $3$-grams carry information about larger $n$-grams:
the count of a single $3$-gram in $\mathbf{S}$ is the
{\em sum} of counts for all $n$-grams with that $3$-gram as the prefix.
So, for instance, if a $3$-gram $x_i$ is in the top-$k$ $3$-grams,
then it is intuitively likely that some $4$-gram whose prefix is $x_i$
will {\bf also} be a top-$k$ $4$-gram.

Empirically validating our intuition is easy:
on a dataset of 18GB of executable programs,
we compute true $3$-grams and true $4$-grams.
For a few values of $k$, we plot
the percentage of $4$-grams that are in the top-$zk$ $3$-grams,
where $z$ is a constant `oversampling' factor.
The results are shown in Fig.~\ref{fig:topk-pct}.
The figure's implication is clear:
if we compute the top-$zk$ $n$-grams,
then we can compute the top-$k$ $(n+1)$-grams
using only the top-$zk$ $n$-grams as prefixes,
which significantly reduces the size of the $n$-grams that must be counted.

\section{Intergrams: Intermediate $n$-gramming}
\label{sec:intergrams}

We can use the $3$-gram algorithm of the previous section as a
starting point for a fast $n$-gramming algorithm
that we call {\bf Intergrams},
following the naming of the programming language INTERCAL~\citep{woods1973intercal}.
The strategy is simple:
suppose that we first computed the top-$zk$ $3$-grams.
Then, we would compute the counts of all $4$-grams whose $3$-gram prefixes were
in the top-$zk$ $3$-grams.
Following this, we would compute the top-$zk$ $4$-grams and
repeat the process for $5$-grams, $6$-grams, etc.
This turns out to be a very effective strategy that
meets our two hardware-based desiderata,
and the strategy is formalized in Algorithm~\ref{alg:intergrams}.

\begin{algorithm}[t]
\small
\begin{center}
\begin{algorithmic}[1]
\Require sequence set $\mathbf{S}$, $n$, $k$

\medskip

\State $P^{(3)} \gets \mathrm{top-}zk\; 3\mathrm{-grams}$ using Algorithm~\ref{alg:gram3_fast}
\ForAll{$j \in [4, n]$}
    \State $C \gets \mathbf{0}^{256zk}$
    \ForAll{sequence $S_i$ in $\mathbf{S}$ {\bf in parallel}}
        \State $C_i \gets \mathbf{0}^{256zk}$
        \ForAll{$j$-gram $x \in S_i$}
            \If{$x$ has $(j - 1)$-gram prefix in $P^{(j - 1)}$}
                \State $C_i[x] \gets 1$
                \Comment{Random access into bit vector.}
            \EndIf
        \EndFor
        \State $C \gets C + C_i$
        \Comment{Sequential flush.}
    \EndFor

    \medskip
    \If{$j < n$}
        \State $P^{(j)} \gets zk\operatorname{-argmax}_{x \in C} C[x]$
        \Comment{Compute prefixes for the next round.}
    \EndIf
\EndFor

\noindent \State {\bf return} $k\operatorname{-argmax}_{x \in C} C[x]$
\end{algorithmic}
\end{center}\caption{{\bf Intergrams}: a fast algorithm to compute the top-$k$ $n$-grams on byte sequences,
using $c_{\mathds{1}}(\cdot, \cdot)$ (e.g. counting an $n$-gram once per sequence only).}
\label{alg:intergrams}
\end{algorithm}

Importantly,
at each stage of the algorithm we keep only $zk$ prefixes;
this means in the next pass we only count $256zk$ $n$-grams.
The size of each bit vector $C_i$ will now be $32zk$ bytes;
this still trivially fits in most processor caches for even large values of $k$ and $z$,
and thus the random access pattern into $C_i$ is
not a problem on subsequent passes.

When $zk$ is less than $65536$,
the memory required for the next iteration's counts array $C$ is less than
the $64$ MB needed for $3$-gram computation,
meaning that it too comfortably fits into the processor cache.
On some systems, $zk$ can be significantly larger before overflowing the cache.
Even if $C$ does not fit fully into the cache,
the sequential access pattern of our flushes will still ensure
high memory bandwidth.

It is worth considering the operation of determining whether a prefix is in $P$.
Of course,
simple iteration over all prefixes in $P$ is very inefficient;
similarly,
we found that even binary search over a pre-sorted $P$ was too slow.
Instead,
a trie structure~\citep{fredkin1960trie} should be used for
fast lookup of a prefix's membership in $P$.
Tries are best described visually: Fig.~\ref{fig:trie} gives an example on byte sequence data.

\begin{figure}[t!]
\begin{center}
\begin{adjustbox}{width=0.75\columnwidth}
\begin{tikzpicture}[auto,>=latex]

    \tikzstyle{rblock}=[draw, shape=rectangle,rounded corners=0.2em, align=center];

    \node[rblock] (root) at (0.0, 0.0) { {\it root} };

    \node[rblock] (l1_x00) at (2.5, 0.0) { {\tt 0x00} };
    \node[rblock] (l1_xff) at (2.5, -3.0) { {\tt 0xff} };

    \node[rblock] (l2_x00) at (5.0, 0.0) { {\tt 0x00} };
    \node[rblock] (l2_x0f) at (5.0, -1.2) { {\tt 0x0f} };
    \node[rblock] (l2_xff) at (5.0, -2.4) { {\tt 0xff} };

    \node[rblock] (l21_x00) at (5.0, -3.0) { {\tt 0x00} };

    \node[rblock] (l3_x00) at (7.5, 0.0) { {\tt 0x00} };
    \node[rblock] (l3_x01) at (7.5, -0.6) { {\tt 0x01} };
    
    \node[rblock] (l3_x11) at (7.5, -1.2) { {\tt 0x11} };
    \node[rblock] (l3_x80) at (7.5, -1.8) { {\tt 0x80} };
    
    \node[rblock] (l3_xff) at (7.5, -2.4) { {\tt 0xff} };

    \node[rblock] (l3_xfe) at (7.5, -3.0) { {\tt 0xfe} };

    \draw [arrow,->] (root.east) to [out=0,in=180] (l1_x00.west);
    \draw [arrow,->] (root.east) to [out=0,in=180] (l1_xff.west);

    \draw [arrow] (l1_x00.east) to [out=0,in=180] (l2_x00.west);
    \draw [arrow] (l1_x00.east) to [out=0,in=180] (l2_x0f.west);
    \draw [arrow] (l1_x00.east) to [out=-10,in=180] (l2_xff.west);

    \draw [arrow] (l1_xff.east) to [out=0,in=180] (l21_x00.west);

    \draw [arrow] (l2_x00.east) to [out=0,in=180] (l3_x00.west);
    \draw [arrow] (l2_x00.east) to [out=0,in=180] (l3_x01.west);

    \draw [arrow] (l2_x0f.east) to [out=0,in=180] (l3_x11.west);
    \draw [arrow] (l2_x0f.east) to [out=0,in=180] (l3_x80.west);

    \draw [arrow] (l2_xff.east) to [out=0,in=180] (l3_xff.west);

    \draw [arrow] (l21_x00.east) to [out=0,in=180] (l3_xfe.west);
\end{tikzpicture}
\end{adjustbox}
\end{center}
\vspace*{-1em}
\caption{An example trie built on length-3 byte sequence data.
Determining whether a prefix $x$ is in $P$
can be done by iterating from the root of the trie
to a leaf,
considering each element of the prefix in succession.
In this trie, we can determine that the byte sequence {\small \tt 0x013300} is not in $P$
at the first iteration.
Other sequences, like {\small \tt 0x000aaf} and {\small \tt 0xff00ff} require more iterations;
sequences like {\small \tt 0x000000} are in $P$ and will reach a leaf during iteration of the trie.}
\label{fig:trie}
\end{figure}
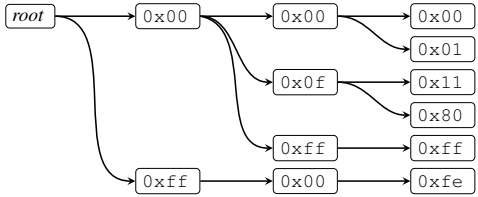

As lookups in this trie are the innermost computation of our entire algorithm,
performance is paramount,
and for this we must again consider the hardware.
But we cannot make lookups sequential in a trie;
we will be jumping from node to node depending on the prefix we are looking up.
So our memory access pattern will appear more random than sequential.
Thus, keeping the size of the trie data structure small is very important;
if it is too large, it will not fit in processor cache
or it will force parts of $C$ or $C_i$ to be evicted from cache.
For byte sequence data,
each trie node can have a maximum of 256 children;
so we can use an 8-bit unsigned integer to track the number of children of a node.
Each node must also hold a value;
an 8-bit unsigned integer suffices.

During the lookup process,
we must find the correct child of the trie to visit
(or if there is no child to visit next, then we terminate early).
To prevent a potentially long iteration over children when a node has many children,
we can use a lookup table.
For byte data, this table needs only 256 elements.
We found lookup tables to be faster
despite their increased memory usage
for nodes with as few as 4 children.

Lastly,
to reduce the number of `faraway' memory accesses,
the memory layout of the trie's children can be ordered by frequency.
For example, if the prefix {\small \tt 0xBEEF00}
is the most common prefix in $\mathbf{S}$,
then the trie nodes corresponding to {\small \tt 0xBE}, {\small \tt 0xBEEF}, and {\small \tt 0xBEEF00}
can be laid out contiguously.
The second-most common child of {\small \tt 0xBEEF}
can then be placed directly after {\small \tt 0xBEEF00},
and so forth.
This means that the most commonly seen prefixes
are more likely to have their trie elements already in the cache,
and less likely to produce cache misses and page faults.

In our implementation,
we were able to restrict the size of the trie to
only $2$ bytes per trie node
plus the memory required for holding the locations of each trie's children
(up to $12$ bytes if the node has less than four children,
and $1$ KB for the lookup table otherwise;
most nodes do not have lookup tables).
The number of nodes in a trie is linear in the number of prefixes,
so with only $zk$ prefixes,
the overall size of the trie is small compared to $C$.
With these optimizations, overall Intergrams runtime was reduced by about 10--15\%.

\section{Theory: Intermediate Filter Recall}
\label{sec:theory}

Next
we aim to validate our claim that the top-$zk$ $n$-grams can effectively filter for the
top-$k$ $(n + 1)$-grams.
For this analysis,
we assume that $n$-grams and $(n+1)$-grams over $\mathcal{S}$ follow a Zipfian distribution with parameter $a$.
That is,
if the $n$-grams $x_1, \ldots, x_{|\mathbf{D}|}$ are ordered by decreasing probability
$p_1, \ldots, p_{|\mathbf{D}|}$, then $p_i \propto \frac{1}{i^a}$.
This assumption is actually pessimistic;
in general, empirical data is such that the parameter $a$ is
greater for $(n+1)$-grams than for $n$-grams.

Although Intergrams does not explicitly search for the exact top-$k$ $(n+1)$-grams over $\mathcal{S}$,
we are able to show that with high probability,
most of the probability mass of the true top-$k$ $(n+1)$-grams is retained.
To show this we first consider the ideal case where the observed $n$-gram counts exactly match the underlying Zipfian frequencies (e.g., no sampling noise).
We will use the following facts about the Zipf distribution.
\begin{obs} \label{obs:zipf}
    Let $M_k \coloneqq \sum_{i=1}^k \frac{1}{i^a}$, and let $|\mathbf{D}|$ be the maximum number of distinct $n$-grams.
    \begin{enumerate}
        \item The probability of sampling $n$-gram $x_i$ is $p_i = \frac{i^{-a}}{M_{|\mathbf{D}|}}$.
        \item The probability of an independently sampled $n$-gram being in the top-$k$ most likely $n$-grams is $M_k/M_{|\mathbf{D}|}$.
    \end{enumerate}
\end{obs}

We observe that for each sequence $s_i \in \mathbf{S}$, each $n$-gram (except for the final one) is the prefix of an $(n+1)$-gram.
Given $|\mathbf{S}|=m$ such sequences and $N$ collected $n$-grams across $\mathbf{S}$, then these are the prefixes of $N-m$ corresponding $(n+1)$-grams.
By this reasoning,
if the top-$zk$ $n$-grams account for some proportion $\beta$ of all $n$-grams, then the proportion of $(n+1)$-grams prefixed by a top-$zk$ $n$-gram is roughly $\beta$ as well:

\begin{restatable}{lem}{betalemma}
\label{lem:beta_to_beta'}
    Suppose the top-$k'$ $n$-grams account for proportion $\beta$ of the total $n$-grams over dataset $\mathbf{S}$.
    Then the $(n+1)$-grams prefixed by one of these top-$k'$ $n$-grams account for at least $\beta' \coloneqq \beta - \frac{m}{N-m}$ of the total $(n+1)$-grams,
    with $N$ the total $n$-grams over $\mathbf{S}$ and $m$ the number of sequences.
\end{restatable}

In the worst case scenario for Intergrams, the top-$zk$ $n$-grams prefix only trivial $(n+1)$-grams, dispersing the mass $\beta$ among less frequent $(n+1)$-grams.
Yet, if $\beta$ is large enough, then even in this adversarial situation we may still collect most frequent $(n+1)$-grams.
Let $X^{n}_k$ be the set of $(n+1)$-grams prefixed by a top-$zk$ $n$-gram.
Our goal is to lower bound the occurrence of $(n+1)$-grams within $X^{n}_k$ that are also among the top-$k$ $(n+1)$-grams.

Specifically, let $u$ be the index of the most frequent $(n+1)$-gram within $X^{n}_k$.
Then the probability mass between $x_u$ and $x_{k}$, 
the $k$-th most frequent $(n+1)$-gram,
lower bounds the relevant mass of $(n+1)$-grams that Intergrams will retain in the intermediate round.

\begin{restatable}{thm}{probthm}
\label{thm:prob_n+1}
    Suppose that $n$-grams and $(n+1)$-grams exactly follow the Zipf distribution (no sampling noise) with parameter $a\neq 1$.\footnote{This can be re-derived for $a=1$ and is conceptually the same.}
    If Intergrams keeps the top $k'$ $n$-grams from the previous pass,
    then (1) the most frequent $(n+1)$-gram it finds will be at least the $u$-th most frequent actual $(n+1)$-gram, where 
    \begin{equation} \label{eq:u_bound}
        u \leq \Big(\big(|\mathbf{D}_{n+1}|^{1-a} - a\big)(1-\beta') + 1\Big)^{\frac{1}{1-a}} - 1.
    \end{equation}

    (2) Furthermore, the $j$-th most frequent $(n+1)$-gram found will be at least the $(u + j - 1)$-th most frequent. 
    
    (3) Finally, the fraction of top-$k$ $(n+1)$-grams recalled by Intergrams is at least
    \begin{equation} \label{eq:prob_frac_bound}
        1 - \frac{(|\mathbf{D}_{n+1}|^{1-a} - a)(1-\beta') - a}{(k + 1)^{1-a} - 1}.
    \end{equation}
\end{restatable}

Intuitively, the faster the tail of the distribution decays, the higher proportion the top-$k$ represents.
When $a$ is not too tiny, the top-$k$ $n$-grams can represent a majority of observed counts,
which forces the majority of $(n+1)$-grams to come from these prefixes.
For larger $a$, such as $a>1$, the bound Eq.~\ref{eq:prob_frac_bound} can be algebraically inverted to yield
asymptotic behavior of the form $1 - \mathcal{O}((k / |D_{n+1}|)^{a - 1})$,
which quickly converges toward successful results for Intergrams.

Now, consider the case where sampling error is present:
henceforth, the counts over $s_i\sim \mathcal{S}$ follow a multinomial generation model with $N$ observations and parameters $\{p_i\}_{i=1}^{|\mathbf{D}|}$,
with empirical probabilities $\hat{p}_i \coloneqq c_i / N$.
While w.l.o.g. the true probabilities follow $p_1 \geq \ldots \geq p_{|\mathbf{D}|}$,
we use the notation $p_{(i)}$ to order by empirical probabilities, so that
$\hat{p}_{(1)} \geq \ldots \geq\hat{p}_{(|\mathbf{D}|)}$. 
Although easy concentration bounds exist when the top-$k$ bins are fixed,
in our case the empirical top-$k$ bins are random.
Even so we may establish bounds.

\begin{restatable}{lem}{concentrationlem}
\label{lem:mult_concentration}
 Choose any $\delta > 0$; with probability at least $1 - \delta$,
\begin{equation}
    \Big| \sum_{i=1}^k \hat{p}_{(i)} - \sum_{i=1}^k p_i \Big| \leq \Delta(\delta) \colonequals 4 \sqrt{\frac{k^2\ln(2|\mathbf{D_n}|/\delta)}{2N}}.
\end{equation}
\end{restatable}

Given that the top-$zk$ $n$-grams are selected empirically over $\mathbf{S}$,
it follows from Lemma~\ref{lem:beta_to_beta'} and Lemma~\ref{lem:mult_concentration} that with probability $1-\delta$ these account for $\beta'' \coloneqq \beta' - \Delta(\delta)$ of the underlying probability measure.
Thus, Theorem~\ref{thm:prob_n+1} can be easily adapted for the noise case, as follows.

\begin{thm} \label{thm:bound_with_noise}
    Suppose the observed $n$-grams are ranked by empirical count and the top $k'$ selected.
    Let $\beta'' = \beta - \frac{m}{N-m} - \Delta(\delta)$.
    Then with probability $1-\delta$, equations (\ref{eq:u_bound}) and (\ref{eq:prob_frac_bound}) hold with $\beta'$ replaced with $\beta''$.
\end{thm}

By setting $k'=zk$ in each result,
we can handle the oversampling value $z$. While for byte sequences the maximum value for $|\mathbf{D}_n|$ is $256^n$,
the bound in Lemma~\ref{lem:mult_concentration} can be strengthened by filtering out prefixes which are not observed in each prior step. 

\section{Experiments}
\label{sec:experiments}

\begin{table}[b!]
\begin{center}
\adjustbox{max width=0.49\textwidth}{
\begin{tabular}{lccc}
\toprule
{\bf Dataset} & {\bf Type} & {\bf Size} & {\bf $m$} \\
\midrule
EMBER\footnote{Work complete before EMBER2024 was read~\cite{joyce_ember2024_2025}}~\cite{anderson2018ember} & bytes & 1009 GB & 800k \\
C4~\cite{raffel2020exploring} & text & 751 GB & 6.22M \\
1000gp~\cite{clarke20121000} & genomics & 1.4 TB & 1.58M \\
\bottomrule
\end{tabular}
}
\end{center}
\vspace*{-1.0em}
\caption{Datasets used in our experiments.  The {\it C4} and {\it 1000gp} datasets were chunked
such that each sequence contained a few hundred to a few thousand lines of source data.}
\label{tab:datasets}
\end{table}

\begin{table*}[t!]
\begin{center}
\begin{tabular}{lccccccccc}
\toprule
{\bf Algorithm} & \multicolumn{3}{c}{{\bf EMBER} ($k = 100$k)} & \multicolumn{3}{c}{{\bf c4} ($k = 10$k)} & \multicolumn{3}{c}{{\bf 1000gp} ($k = 10$k)} \\
 & Runtime & Speedup & Jaccard & Runtime & Speedup & Jaccard & Runtime & Speedup & Jaccard \\
\midrule
{\small \tt hg-vanilla}            & 9078.5s & --          & --   & 39787.1s & --          & --        & 10042.0s & --          & --        \\
{\small \tt hg-cuckoo}             & 8987.9s & 1.01x       & --   & 37506.8s & 1.06x       & --        & 9682.3s  & 1.04x       & --        \\
{\small \tt hg-largepage}          & 8394.7s & 1.08x       & --   & 37214.8s & 1.06x       & --        & 8640.6s  & 1.16x       & --        \\
{\small \tt hg-trie}               & 8286.5s & 1.10x       & --   & 33026.1s & 1.20x       & --        & 8130.2s  & 1.24x       & --        \\
{\small \tt hg-fast}               & 7162.6s & 1.27x       & --   & 30811.2s & 1.29x       & --        & 7061.5s  & 1.42x       & --        \\
{\small \tt intergrams}, $z = 1$   & 1413.7s & {\bf 6.42x} & 0.71 & 1183.9s  & {\bf 33.6x} & 0.91      & 1215.7s  & {\bf 8.26x} & {\bf 1.0} \\
{\small \tt intergrams}, $z = 1.5$ & 1458.8s & {\bf 6.22x} & 0.91 & 1373.7s  & {\bf 29.0x} & {\bf 1.0} & 1152.9s  & {\bf 8.71x} & {\bf 1.0} \\
{\small \tt intergrams}, $z = 2$   & 1771.4s & {\bf 5.13x} & 0.97 & 1501.9s  & {\bf 26.5x} & {\bf 1.0} & 1144.6s  & {\bf 8.77x} & {\bf 1.0} \\
\bottomrule
\end{tabular}
\end{center}
\vspace*{-1em}
\caption{Runtime results for Intergrams and variants of the hash-gramming approach.
In every case the Intergrams algorithm is able to provide significant speedup (up to 30x!),
while still recovering nearly every true top-$k$ $n$-gram (sometimes exactly!).}
\label{tab:count_one_experiments}
\end{table*}

We now demonstrate the superiority of the Intergrams algorithm
as compared to other approaches,
validating each of our hardware-based observations along the way.
As hash-gramming is orders of magnitude faster than any other algorithms~\cite{raff2018hash},
we compare only against hash-gramming with
the `small tricks' discussed earlier. 

Our experimental system is a very powerful single system
equipped with 4 32-core Xeon E7-8867 v3 processors,
for a total of 128 cores.
Each core has 32 KB L1 cache and 512 KB L2 cache.
Each processor shares its 45 MB L3 cache across cores,
giving a total of 180 MB of L3 cache.
The system is equipped with 2 TB of DDR3-1333 RAM,
with a disk controller capable of serving data from a RAID at $\sim\!5$ GB/s.

We implemented the Intergrams algorithm carefully in C++,
using the Armadillo library for timing~\cite{sanderson2016armadillo}.
We process sequences in parallel using thread pairs:
one thread reads the sequence from disk
(this thread is mostly idle, waiting on I/O)
while another thread processes data that the other has buffered.
This thread pair strategy is necessary because Intergrams
is bottlenecked by disk;
thus, a thread entirely for I/O ensures
maximum disk throughput.

We also implemented hash-gramming carefully in C++;
however, due to its poor memory access patterns,
it cannot saturate disk bandwidth
and thus each sequence is assigned to its own thread,
which performs both I/O and the hash-gramming computations.
We applied the performance tricks detailed earlier
to hash-gramming, producing these variants:

\begin{itemize}
    \item {\small \tt hg-vanilla} is hash-gramming implemented as written.
    \item {\small \tt hg-largepage} uses 2 MB page sizes to reduce the TLB misses to the global counts array.
    \item {\small \tt hg-prefetch} uses software prefetching to mask the latency of fetching the global counts array from RAM.
    \item {\small \tt hg-cuckoo} uses cuckoo hashing~\citep{pagh2004cuckoo} to minimize the updates to the global counts array.
    \item {\small \tt hg-fast} uses our tuned trie structure for the second pass, as opposed to the standard C++ {\tt unordered\_map}.
\end{itemize}

The optimizations in each variation are cumulative:
that is, {\small \tt hg-cuckoo} also contains the optimizations from {\small \tt hg-prefetch} and {\small \tt hg-largepage},
and so forth.
As we designed Intergrams with all these optimizations in mind,
ablating them individually does not make sense.
Instead, to control the accuracy of Intergrams,
we take $z \in \{ 1, 1.5, 2 \}$.

For datasets, we used three large datasets from different application areas,
detailed in Table~\ref{tab:datasets}.
Our implementations were tuned to byte data,
and so we processed the text and genomic data at the character (byte) level.
Note that for both plaintext (ASCII) and genomics data,
the size of $|D|$ is smaller than $256^n$,
meaning that additional optimization implementations could be made
(see the trie discussion).

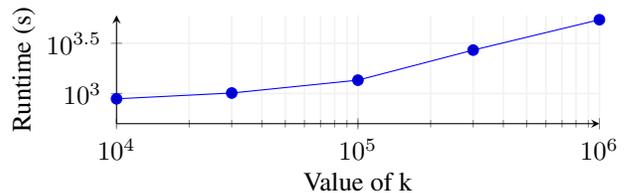
\begin{figure}[!t]
\begin{tikzpicture}
\begin{axis}[
      axis lines=left,
      width=0.45\textwidth,
      height=0.17\textwidth,
      ymin=500,ymax=6000,
      xmin=10000, xmax=1000000,
      xlabel={Value of k},
      ylabel={Runtime (s)},
      xmode=log,
      log basis x={10},
      ymode=log,
      log basis y={10},
      grid=both,grid style={line width=0.6pt, draw=gray!10}
]
    \addplot table [x=k, y=time, col sep=comma, color=blue] {data/k_sweep.csv};
\end{axis}
\end{tikzpicture}
\vspace*{-1.8em}
\caption{Runtime results for a sweep of $k$ on the EMBER dataset.  Y-axis is logarithmic.}
\label{fig:k_sweep_ember}
\end{figure}

\begin{table}[!t]
\begin{center}
\small
\begin{tabular}{lcc}
\toprule
{\bf Step} & {\bf Runtime} & {\bf Throughput} \\
\midrule
$3$-gram pass & 228.2s & 4.42 GB/s \\
top-$zk$ $3$-grams/trie building & 3.25s & -- \\
$4$-gram pass & 205.95s & 4.90 GB/s \\
top-$zk$ $4$-grams/trie building & 0.97s & -- \\
$5$-gram pass & 189.37s & 5.33 GB/s \\
top-$zk$ $5$-grams/trie building & 1.16s & -- \\
$6$-gram pass & 190.57s & 5.29 GB/s \\
top-$k$ $6$-grams & 0.29s & -- \\
\bottomrule
\end{tabular}
\end{center}
\vspace*{-0.65em}
\caption{Runtime breakdown for each step of the Intergrams algorithm
on EMBER with $k$ set to $10$k and $z = 1.5$.}
\label{tab:breakdown}
\end{table}

Table~\ref{tab:count_one_experiments} contains runtime results for each of the datasets.
Intergrams {\em significantly} outperforms all variants of hash-gramming,
even the most optimized.
In the best cases, Intergrams provides 20--30x speedup
due to its hardware-informed design.
In addition, even with $z = 1$, it returns $n$-grams that match those
returned by hash-gramming.

Next, we consider the effect of increasing $k$ on the runtime.
Figure~\ref{fig:k_sweep_ember} shows each algorithm's runtime on the EMBER dataset
when $k$ is swept from $1$k to $1$M.
Jaccard similarities were comparable to what is seen in Table~\ref{tab:count_one_experiments}.

Lastly,
to show the relative cost of each step of the Intergrams algorithm,
Table~\ref{tab:breakdown} shows the breakdown of each step of the Intergrams algorithm on
each of the three datasets.
Since no trie lookup is necessary for the $3$-gram pass,
that step is the most efficient and can maximize disk throughput.
The top-$zk$ steps and trie building steps are negligible compared to the time it takes to pass over the data.
In later passes of the algorithm,
the trie gets more selective (as compared to $|D|$)
and this phenomenon accounts for the $5$-gram and $6$-gram passes being faster than the $4$-gram pass.

\section{Conclusion}

We introduced the Intergrams algorithm,
a fast algorithm for computing the top-$k$ $n$-grams
on very large datasets.
Its empirical performance beats all known other algorithms,
and it has favorable theoretical guarantees.
Code can be found at {\small \tt github.com/rcurtin/Intergrams}.
We are exploring further improvement by
exploiting the Zipf distribution to skip data
for a slight loss in accuracy,
as in~\citet{raff2025zipf}.

\bibliography{refs}

\newpage
\appendix

\section{Proofs}

\betalemma*
\begin{proof}
    Take the top $k$ such $n$-grams and denote their total occurrences over $\mathbf{S}$ as $c_{k}$. Now consider each sequence $s_i$ separately.
If the final $n$-gram is one of the top $k$,
then that one does not prefix an $(n+1)$-gram.
Otherwise each of the top $k$ $n$-gram occurrences is a prefix.
Therefore, the total occurrences of $
(n+1)$-grams prefixed by a top $k$ $n$-gram is between $c_{k}-m$ and $c_{k}$.

By this reasoning, we observe that if the top $k$ $n$-grams account for some fraction $\beta \coloneqq c_{k}/N $ of all the $n$-grams over $\mathbf{S}$,
then these $n$-grams prefix at least $(c_{k} - m) / (N - m)$ of the total mass of observed $(n+1)$-grams. The conclusion follows immediately.
\end{proof}

\probthm*
\begin{proof}
    (1) Assume there are $|\mathbf{D}_n|$ and $|\mathbf{D}_{n+1}|$ unique $n$-grams and $(n+1)$-grams respectively.
    Assuming the vocabulary size is 256,
    the observed top $k$ form the prefixes of at most $256k$ candidate $(n+1)$-grams.
    From Obs.~\ref{obs:zipf} and assuming no sampling noise, we know that $\beta \coloneqq M_k / M_{|\mathbf{D}|}$.
    The mass over $n$-grams transfers to the $(n+1)$-grams as described above,
    so then
$$(1- M_u/M_{|\mathbf{D}_{n+1}|}) -(1-M_{u+256k} /M_{|\mathbf{D}_{n+1}|}) \geq \beta' $$

To simplify the math we relax the constraint that there are $256k$ candidates, which means our estimate of $u$ will be bounded by:
\begin{equation} \label{eq:beta_transfer}
    M_u/M_{|\mathbf{D}_{n+1}|} \leq 1-\beta'
\end{equation}

That is, we want to find an upper bound on $u$ such that the expression holds.
To make the expression tractable, we lower bound it with the integral approximation of $M_k$ given by Lemma~\ref{lem:integral_bound} to get:
$$ \frac{(u + 1)^{1-a} - 1}{|\mathbf{D}_{n+1}|^{1-a} - a} \leq 1 - \beta'$$
When $a < 1$ the denominator is positive. Multiplying by the denominator and rearranging gives the result
$$u \leq \Big(\big(|\mathbf{D}_{n+1}|^{1-a} - a\big)(1-\beta') + 1\Big)^{\frac{1}{1-a}} - 1$$

When $a > 1$, both numerator and denominator are negative. Then multiplication by the denominator reverses the inequality.
$$ (u+1)^{1-a} \geq (1-\beta') (|\mathbf{D}_{n+1}|^{1-a} - a)+ 1$$
Since $1-a<0$, $(u+1)^{1-a} = 1/(u+1)^{a-1}$.
Thus we can take the reciprocal of both sides and rearrange to arrive at the same result.

(2) To show this, we observe that the worst case shown in Eq.~\ref{eq:beta_transfer} actually assumes that all $(n+1)$-grams from the $u$-th onward are recalled. In such a case, the 2nd most frequent $(n+1)$-gram would be position $u+1$, and so on. In reality more favorable $(n+1)$-grams will be found, but that only decreases the index value $u+j-1$, so the upper bound holds.

(3) Finally, we want to bound the actual probability mass of the top $k$ $(n+1)$-grams which are captured by the intermediate step.
The cumulative probability mass up to $u$ and the cutoff $k$ are $\sum_{i=1}^u p_i$ and $\sum_{i=1}^k p_i$, respectively.
We thus define the quantity of interest as
\begin{align}
    \frac{\sum_{i=1}^k p_i - \sum_{i=1}^u p_i}{ \sum_{i=1}^k p_i } &= 1 - \frac{\sum_{i=1}^u p_i}{\sum_{i=1}^k p_i} \\
    &= 1 - \frac{M_u}{M_k}\\
    &\geq 1 - \frac{u^{1-a} - a}{(k+1)^{1-a} - 1}
\end{align}
We can further substitute the bound for $u$ and obtain
$$1 - \frac{u^{1-a} - a}{(k+1)^{1-a} - 1} \geq 1 - \frac{(|\mathbf{D}_{n+1}|^{1-a} - a)(1-\beta') - a}{(k + 1)^{1-a} - 1}$$    
\end{proof}

\concentrationlem*
\begin{proof}
    If the top $k$ $n$-grams by empirical count are the same as the actual top $k$, then standard bounds can be applied. However, in practice, the actual $n$-grams selected will differ.
    So there are two sources of randomness: (1) the deviation of empirical counts around the mean among the selected bins, and (2) the deviation of the selected high-frequency bins around the true high-frequency bins.
    Using this information we can decompose the sum as
    \begin{align*}
        \bigg| \sum_{i=1}^k \hat{p}_{(i)} &- \sum_{i=1}^k p_i \bigg| \\
        &= \bigg| \sum_{i=1}^k \hat{p}_{(i)}  + \sum_{i=1}^k p_{(i)} - \sum_{i=1}^k p_{(i)} - \sum_{i=1}^k p_i  \bigg|\\
        &\leq \bigg| \sum_{i=1}^k \hat{p}_{(i)}  - \sum_{i=1}^k p_{(i)} \bigg| + \bigg| \sum_{i=1}^k p_i  - \sum_{i=1}^k p_{(i)} \bigg|
    \end{align*}
    
    Focusing on the second term,
    we observe the summations are over two sets of indices $K$ and $\hat{K}$,
    consisting of the top $k$ actual and top $k$ empirical $n$-grams, respectively.
    For any $n$-grams in the intersection of the sets, the contribution to the sum is 0.
    Thus any remaining terms are in the set difference $K \Delta \hat{K}$.
    We pair these remaining terms from each set in descending order, denoting the set of pairs as $J \coloneqq (K\setminus \hat{K}) \times  (\hat{K} \setminus K)$:

    \begin{align}
        \bigg| \sum_{i=1}^k p_i  - \sum_{i=1}^k p_{(i)} \bigg|  &= \bigg| \sum_{i\in K} p_i  - \sum_{j\in \hat{K}} p_j \bigg| \\
        &= \bigg| \sum_{i \in K\setminus\hat{K}} p_i  - \sum_{j\in \hat{K}\setminus K} p_j\bigg|\\
        &= \bigg| \sum_{(i, j)\in J} (p_i - p_j) \bigg|
    \end{align}

    Now we can use a key observation: for each pair $(i, j)\in J$, $p_i \geq p_j$ but $\hat{p}_i \leq \hat{p}_j$.
    This can be seen as the first index consists of true top $k$ probabilities that did not make the top $k$ empirical probabilities,
    and vice versa for the second index.
    Furthermore, there are no overlapping $n$-grams.    
    Then
    \begin{align}
        \bigg| \sum_{(i, j)\in J} (p_i - p_j) \bigg| &= \bigg| \sum_{(i, j)\in J} (p_i - p_{j} + \hat{p}_i - \hat{p}_i)\bigg|\\
        &\leq \bigg| \sum_{(i, j)\in J} (p_i - p_{j} + \hat{p}_{j} - \hat{p}_i)\bigg|\\
        &=  \bigg| \sum_{i \in K \setminus \hat{K}} (p_i - \hat{p}_i) + \sum_{j\in \hat{K} \setminus K} (p_{j} - \hat{p}_{j} )\bigg| \\
        &\leq  \sum_{i \in K \setminus \hat{K}} |p_i - \hat{p}_i| + \sum_{j\in \hat{K} \setminus K} |p_{j} - \hat{p}_{j} | \\
        &\leq \sum_{i \in K\cup \hat{K}} |p_i - \hat{p}_i| \end{align}

Now we return to the first term $\big| \sum_{i=1}^k \hat{p}_{(i)}  - \sum_{i=1}^k p_{(i)} \big|$.
We can also rewrite this as
\begin{align}
    \bigg| \sum_{i=1}^k \hat{p}_{(i)}  - \sum_{i=1}^k p_{(i)} \bigg| &= \bigg| \sum_{i\in \hat{K}} (\hat{p}_{i}  - p_{i}) \bigg|\\
    &\leq \sum_{i\in{\hat{K}}} |\hat{p}_i - p_i|\\
    &\leq \sum_{i\in K \cup \hat{K}} |\hat{p}_i - p_i|
\end{align}
So now we have
$$\bigg| \sum_{i=1}^k \hat{p}_{(i)}  - \sum_{i=1}^k p_{i} \bigg| \leq 2 \sum_{i\in K \cup \hat{K}} |\hat{p}_i - p_i|$$

While it may appear that standard concentration bounds can apply,
the issue is that the summation is over a random set $\hat{K}$, which makes the quantities dependent.
In particular, conditioning on the fact that $\hat{K}$ are the largest empirical counts, the deviations are likely to be particularly large.
To address this, we work with the maximal deviation
$$\sum_{i\in K \cup \hat{K}} |\hat{p}_i - p_i| \leq 2k \max_{i\in 1 \ldots |\mathbf{D}|} |\hat{p}_i - p_i|$$
since each term is bounded by the maximal deviation and there are at most $2k$ terms in the summation.

For any given fixed $n$-gram $x_i$, Hoeffding's inequality gives $P(|\hat{p}_i - p_i| \geq \epsilon) \leq 2 \exp( -2N \epsilon^2 )$.
By the union bound, then
$$P(\max_i |\hat{p}_i - p_i| \geq \epsilon) \leq 2|\mathbf{D}|\exp(-2N\epsilon^2)$$
Setting the right hand side to $\delta$, then with probability at least $1-\delta$, $$\max_i |\hat{p}_i - p_i| \leq \sqrt{\frac{\ln(2|\mathbf{D}|/\delta)}{2N}}$$

Combining these together, we conclude that
$$ \bigg| \sum_{i=1}^k \hat{p}_{(i)} - \sum_{i=1}^k p_i \bigg| \leq 4k \sqrt{\frac{\ln(2|\mathbf{D}|/\delta)}{2N}}$$

\end{proof}

\begin{lem} \label{lem:integral_bound}
    If $a \neq 1$, then
    $$ \frac{(k + 1)^{1-a} - 1}{1-a} \leq M_k \leq \frac{k^{{1-a}} - a}{1-a}$$
    and if $a = 1$, then
    $$ \ln(k + 1) \leq M_k \leq \ln(k) + 1  $$
\end{lem}
\begin{proof}
    $M_k$ is a decreasing function of $k$.
    Considering $M_k$ as a set of rectangles with width 1 and height $1/x^a$,
    we see immediately the integral bound
    $$M_k \geq \int_{1}^{k+1} \frac{1}{x^a} dx$$

    Shifting the rectangles one unit to the left, we then see
    $$\sum_{x=2}^k \frac{1}{x^a} \leq \int_1^k \frac{1}{x^a} dx$$

    Combining we get
    $$\int_1^{k+1} \frac{1}{x^a} dx \leq \sum_{x=1}^k \frac{1}{x^a} \leq \int_1^k \frac{1}{x^a}dx + 1$$

    Solving the definite integrals gives the final result.
\end{proof}

\end{document}